\documentclass{article}
\usepackage{spconf,amsmath,graphicx}

\title{Adaptive Predictive Sampling and Communication \\ for Real-time Monitoring}

\name{Erfan Delfani and Nikolaos Pappas}
\address{Department of Computer and Information Science, Linköping University, Linköping, Sweden}

\usepackage{graphicx} 
\usepackage{amsmath}
\usepackage{amssymb}
\usepackage{mathtools}
\usepackage{xcolor}
\usepackage[table]{xcolor}
\usepackage{float}
\usepackage{multirow}
\usepackage{booktabs}
\usepackage{caption}
\usepackage{subcaption}
\usepackage{microtype}
\usepackage{verbatim}

\usepackage{amsthm}
\newtheorem{theorem}{Theorem}
\newtheorem{lemma}{Lemma}

\begin{document}
	
    \setlength{\textfloatsep}{6pt plus 2pt minus 2pt}
    \setlength{\floatsep}{6pt plus 2pt minus 2pt}
    \setlength{\intextsep}{6pt plus 2pt minus 2pt}
    
    \setlength{\abovedisplayskip}{4pt plus 2pt minus 2pt}
    \setlength{\belowdisplayskip}{4pt plus 2pt minus 2pt}

    \captionsetup[figure]{skip=3pt}
    \captionsetup[table]{position=above, skip=3pt}
	
	\maketitle
	
	\begin{abstract}   
        We propose an adaptive predictive sampling and communication framework for real-time remote monitoring over packet-erasure channels. Rather than relying on a prescribed dynamical model or continuous sensing, our model-free method uses Taylor-expansion-based prediction to adaptively estimate local signal dynamics and proactively schedule the next sampling times. An analytical compensation factor accounts for packet losses to satisfy a prescribed reconstruction-error level. We develop adaptive sample-wise and episodic policies and evaluate their performance against event-driven and uniform sampling and communication baselines.
	\end{abstract}
	
	\section{Introduction}
	
	Real-time remote monitoring in Internet of Things (IoT) and cyber-physical systems requires resource-constrained sensors to track continuous physical processes and deliver timely updates to a remote destination, where a real-time estimator reconstructs the monitored signal. Standard uniform sampling is often inefficient because its rate must accommodate the fastest signal variations. When the process evolves slowly, this fixed rate generates redundant samples and wastes communication bandwidth and sensor energy \cite{AstromBernhardsson2002,Miskowicz2006,kutsevol2023experimental}.

    Event-driven sampling reduces such redundancy by transmitting only when the discrepancy between the current signal and the receiver's estimate exceeds a prescribed threshold \cite{AstromBernhardsson2002,Miskowicz2006,Tabuada2007}. However, evaluating this condition still requires fine-grained signal sampling. Self-triggered schemes avoid continuous sensing by determining future sampling instants from previously acquired information \cite{AntaTabuada2010,HeemelsJohanssonTabuada2012}, while locally adaptive schemes adjust sampling intervals according to observed signal dynamics \cite{FeiziGoyalMedard2010,FeiziGoyalMedard2012}. Wireless packet erasure channels further complicate remote reconstruction by delaying update delivery and causing error accumulation \cite{Schenato2008,NetworkedEventBasedDropout2009,PremaratneHalgamugeMareels2013,PengHan2016}.
    Thus, efficient remote monitoring requires sampling and communication decisions that account for both signal dynamics and the end-to-end information path, consistent with recent approaches based on Age of Information (AoI) and semantics-aware communication \cite{YatesSunBrownKaulModianoUlukus2021,kosta2017age,kountouris2021semantics}.
    Many remote estimators, however, assume a prescribed dynamic model, such as Markovian, autoregressive, or state-space models \cite{luo2025information,SunPolyanskiyUysal2020,ornee2021sampling,shisher2024monotonicity}. Predictors such as Kalman filters likewise depend on assumed system parameters, making them less suitable when signal dynamics are unknown. Moreover, existing model-free reconstruction frameworks \cite{FeiziGoyalMedard2010,FeiziGoyalMedard2012} do not consider communication over packet-erasure channels.
	
	To address these gaps, we propose a joint sampling and communication framework for predicting and reconstructing continuous signals without assuming a prescribed dynamical model. Our approach analytically determines the next sampling intervals in advance to achieve a target reconstruction error, based on local signal dynamics that are derived and adaptively updated using previous samples. Specifically, we: 
	(i) develop a predictor based on Taylor expansion and Newton divided differences using the past $N+1$ successfully delivered samples, together with an adaptive filter for tracking local derivative bounds;
	(ii) derive an analytical packet-erasure compensation factor $\mathcal{C}$ that enables proactive sampling and transmission over unreliable channels; and
	(iii) develop two adaptive policies: a sample-wise policy and an episodic policy for scenarios in which per-sample control is infeasible.

	\section{System Model}
	
	We consider a real-time remote monitoring system (Fig. \ref{fig_SystemModel})
	where a source samples a continuous-time signal $x(t)$ and transmits
	measurements to a destination over an unreliable erasure channel. Under packet-based transmission, each transmission succeeds without distortion with probability $p_s$ and fails with probability $1-p_s$.
    An error-free, instantaneous link-layer
	\textsc{ack} informs the source about the outcome of the transmission. The destination reconstructs the signal in real time, and our objective is to minimize sampling and transmission operations while targeting a prescribed peak reconstruction error $\epsilon_d$. Source and destination maintain \emph{identical} real-time
	predictors driven by successfully delivered samples, while a source-side \emph{controller} schedules future
	sampling instances.
	
	The controller operates in either \emph{sample-wise} or \emph{episodic} mode.
Sample-wise control determines only the next sampling time, as in \emph{Event-Driven (\textsc{ed})} and
\emph{adaptive predictive joint sampling and communication (\textsc{apjsc})} policies.
Episodic control jointly determines the sampling schedule over the next block
of $D$ samples, as in the \emph{block-based episodic} policy. \emph{Uniform} is a non-adaptive baseline with a sampling interval that remains fixed over time. Under all policies, the latest sample is transmitted until it is delivered or replaced by a new one (e.g., in \textsc{apjsc}, samples are never dropped because sampling awaits delivery). 
Transmission attempts occur at the beginning of discrete time slots of duration $T_C$.
Let $k \in \{1,2,3,\dots\}$ index successfully delivered samples. We denote by $t'_k$ the reception time of the $k$-th delivered sample, $t^\ast_k$ its generation (sampling) time, and $L_k \in \{1,2,3,\dots\}$ the number of transmissions required for delivery. With sampling and transmission synchronized to slot boundaries and negligible transmission and propagation delays, the $k$-th sample is received at the onset of its successful transmission attempt: $t'_k \!=\! t^\ast_k \!+\! (L_k \!-\! 1)T_C$.
	
	We consider clean signals with negligible measurement noise, as obtained from accurate sensors or appropriate analog filtering, and signals that are differentiable up to order $N+1$ (except at several isolated points; see stepped-frequency signals in Section~\ref{sec_SteppedFreResults}). No prior dynamical model of $x(t)$ is assumed; prediction therefore relies solely on previously collected
	samples. For this model-free signal, we adopt a Taylor-series-based prediction method.

	\begin{figure}[t]
		\centering
		\includegraphics[width=0.98\linewidth]{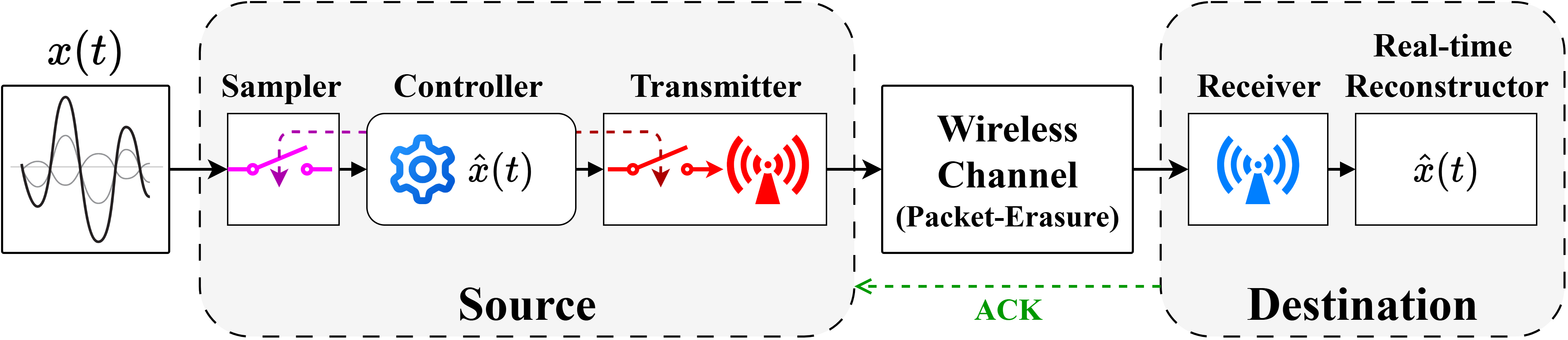}
        \vspace{4pt}
		\caption{System Model.}
		\label{fig_SystemModel}
	\end{figure}

	\textbf{Taylor Expansion for Prediction:}  
    For any $t \ge t'_k$, the Taylor expansion of $x(t)$ around $t^\ast_k$ is given by:
	\begin{align}
		x(t) = \sum_{n=0}^{N} \frac{x^{(n)}(t^\ast_k)}{n!} \Delta^n
		+ \frac{x^{(N+1)}(\xi)}{(N+1)!} \Delta^{N+1},
	\end{align}
	where $\Delta=t-t^\ast_k$ is the \emph{AoI}, i.e., the time elapsed since the generation of the last delivered sample, and $\xi\in(t^\ast_k,t)$ is an unknown point. We obtain the derivatives $x^{(n)}(t^\ast_k)$ from the delivered samples $x(t^\ast_k),x(t^\ast_{k-1}),\ldots,x(t^\ast_{k-N})$ via non-uniform Newton divided differences \cite[Section~2.1.3]{stoer1993introduction}, yielding $x(t) \!=\! \hat{x}(t) \!+\! R_N(t)$, where:
	\begin{gather}
		\hat{x}(t) \!=\! f[t^\ast_k] \!+\! \!\sum_{n=1}^N f[t^\ast_k, t^\ast_{k-1}, \dots, t^\ast_{k-n}] \!\prod_{j=0}^{n-1} (t \!-\! t^\ast_{k-j}), \\
		R_N(t) \!=\! \frac{x^{(N+1)}(\zeta)}{(N\!+\!1)!}  \!\prod_{j=0}^{N} (t \!-\! t^\ast_{k-j}) \!=\! \frac{x^{(N+1)}(\zeta)}{(N\!+\!1)!}  \mathcal{Y}_k(\Delta), \\
		\mathcal{Y}_k(\Delta)
		= \prod_{j=0}^{N}(\Delta+t^\ast_k-t^\ast_{k-j}),
	\end{gather}
	where $f[t^\ast_k]\!=\!x(t^\ast_k)$ and the divided differences are given recursively:
	\begin{align*}
    \begin{array}{cc}
		f[t^\ast_k, t^\ast_{k-1}, \dots, t^\ast_{k-n}] = \frac{f[t^\ast_k, \dots, t^\ast_{k-n+1}] - f[t^\ast_{k-1}, \dots, t^\ast_{k-n}]}{t^\ast_k - t^\ast_{k-n}}.
        \end{array}
	\end{align*}
    
	Here, $\zeta\in(t^*_{k-N},t)$ is an unknown point, $\hat{x}(t)$ denotes the predicted (reconstructed) signal, and $R_N(t)$, which is a polynomial of order $N+1$ in $\Delta$, $R_N(t)=P_{N+1}(\Delta)$, characterizes the \emph{analytical prediction error}, $e_{\textsc{p}}(t)=|P_{N+1}(\Delta)|$. Since $x^{(N+1)}(\zeta)$ is unknown, $e_{\textsc{p}}(t)$ cannot be computed explicitly; in the following sections, we develop methods to adaptively estimate it. The \emph{true reconstruction error} is:
	\begin{align}
		e_\textsc{r}(t) \!=\! \left|x(t)-\hat{x}(t)\right|,
	\end{align}
	which is known only at the sampling times of successfully delivered samples. Specifically,
	\begin{align}
		\epsilon^{\mathrm{true}}_{k+1}
		\!=\! e_\textsc{r}(t^\ast_{k+1})
		\!=\! \left|x(t^\ast_{k+1})-\hat{x}(t^{\ast-}_{k+1})\right|,
		\label{eqn_TrueError_k}
	\end{align}
	where $t^{\ast-}_{k+1}$ denotes the instant immediately prior to $t^\ast_{k+1}$, indicating that $\hat{x}(t^{\ast-}_{k+1})$ is computed using $x(t^\ast_{k}), \ldots, x(t^\ast_{k-N})$.

	\section{Sample-Wise Control}
	\label{sec_SampleWise}
	We first consider a baseline event-driven (hereafter, \textsc{ed}) policy and formulate it for both perfect and unreliable channels. We then propose the \textsc{apjsc} scheme for scenarios in which continuous sampling is infeasible or costly.
    
	\textbf{\textsc{ED}:} 
	Under the standard \textsc{ed} policy, a sample is transmitted only when the
	true reconstruction error exceeds $\epsilon_d$:
	\begin{align}
		t_{k+1} &= t^\ast_k + \Delta_{\mathrm{ED}}, \\
		\Delta_{\mathrm{ED}} &= \inf \left\{ \Delta = m T_C \;\middle|\; m \in \mathbb{N}_0, \; e_\textsc{r}(t^\ast_k + \Delta) \ge \epsilon_d \right\}.
	\end{align}
	
	For an imperfect channel, failed transmissions delay delivery and increase the reconstruction error. We therefore proactively trigger a transmission when the \emph{expected} reconstruction error over channel realizations reaches $\epsilon_d$, i.e., $t_{k+1} \!=\! t^\ast_k + \Delta^\ast_{\mathrm{ED}}$, where:
	\begin{align}
		\Delta^\ast_{\mathrm{ED}} \!=\! \inf \!\left\{ \Delta \!=\! m T_C \;\middle|\; m \!\in\! \mathbb{N}_0, \; \mathbb{E}\!\left[e_\textsc{r}(t^\ast_k \!+\! \Delta)\right] \!\ge\! \epsilon_d \right\}.
	\end{align}
	
	The number of transmission attempts $\mathcal{L}$ is geometrically distributed: $\mathbb{P}(\mathcal{L}\!=\!\ell)=(1\!-\!p_s)^{\ell-1}p_s$. Let $\tau_1 = t^\ast_k + \Delta_0$ denote the initial
	attempt instant, with error $e_\textsc{r}(\tau_1)$. Assuming
	$x^{(N+1)}(\zeta)$ remains approximately constant over the transmission
	slots until delivery (for typical $p_s$), the interpolation error at a
	subsequent attempt $\tau_{\ell} = t^\ast_k + \Delta_0 + (\ell\!-\!1)T_C$ scales as:
	\begin{align}
    \begin{array}{cc}
         e_\textsc{r}(\tau_\ell) \approx
		\prod_{j=0}^{N}\frac{\tau_\ell-t^\ast_{k-j}}
		{\tau_1-t^\ast_{k-j}} e_\textsc{r}(\tau_1).
    \end{array}
		\label{eqn_ch_approx}
	\end{align}
    
	Taking the expectation over the channel realization yields
	$\mathbb{E}\!\left[e_\textsc{r}(t^\ast_k\!+\!\Delta)\right]\!=\!\sum_{\ell=1}^{\infty} \!\mathbb{P}(\mathcal{L}\!=\!\ell) e_\textsc{r}(\tau_\ell) \approx\mathcal{C}e_\textsc{r}(\tau_1)$, where:
	\begin{align}
		\mathcal{C} \!=\!\! \sum_{\ell=1}^{\infty} p_s (1\!-\!p_s)^{\ell-1}
		\prod_{j=0}^{N}\frac{\Delta_0\!+\!(\ell\!-\!1)T_C\!+\!t^\ast_k\!-\!t^\ast_{k-j}}
		{\Delta_0\!+\!t^\ast_k\!-\!t^\ast_{k-j}}.
		\label{eqn_C_coef}
	\end{align}
    
	The triggering condition therefore becomes:
	\begin{align}
		\Delta^\ast_{\mathrm{ED}} \!=\!
		\inf \!\left\{ \Delta_0 \!=\! m T_C \;\middle|\; m \!\in\! \mathbb{N}_0, \;
		e_\textsc{r}(t^\ast_k\!+\!\Delta_0) \!\ge\! \frac{\epsilon_d}{\mathcal{C}} \right\}.
		\label{eqn_ED_C}
	\end{align}

	\textbf{\textsc{APJSC}:} 
	When continuous sensing is infeasible, we propose the \textsc{apjsc} scheme, based on the analytical prediction error. A new sample is acquired when the analytical prediction error $e_{\textsc{p}}(t)=|P_{N+1}(\Delta)|$ reaches $\epsilon_d/\mathcal{C}$ and is transmitted until successful delivery, after which the next sampling time is determined. Thus, $t_{k+1} = t^\ast_k + \Delta^\ast_{\mathrm{APJSC}}$, where:
	\begin{align}
		\Delta^\ast_{\mathrm{APJSC}} \!=\! \inf \!\left\{ \Delta \!=\! m T_C \middle|  \frac{|x^{(N+1)}(\zeta)|}{(N\!+\!1)!}\mathcal{Y}_k(\Delta) \!\ge\! \frac{\epsilon_d}{\mathcal{C}} \right\}\!.
	\end{align}
    
	Here, $x^{(N+1)}(\zeta)$ is the $(N\!+\!1)$-th derivative at an unknown
	$\zeta \in (t^\ast_{k-N},t)$. If a global bound $U_{N+1}$ is known,
	it can replace $|x^{(N+1)}(\zeta)|$, but this yields a conservative criterion
	that fails to capture local signal dynamics. We therefore develop an adaptive
	criterion to update a local estimate of $x^{(N+1)}(\zeta)$.
	
	\textit{Adaptive filter:}
	We parameterize the unknown derivative bound
	$\frac{|x^{(N+1)}(\zeta)|}{(N+1)!}$ by an adaptive coefficient $\alpha_k$,
	yielding the prediction error polynomial
	$e_{\textsc{p}}(t)=\alpha_k\mathcal{Y}_k(\Delta)$.
	We formulate the tracking of this local bound as an adaptive \emph{system identification} problem \cite{haykin2008adaptive}: $\epsilon^{\mathrm{true}}_{k+1}
	= \alpha_k\mathcal{Y}^\ast_k + v_k$,
	where $\mathcal{Y}^\ast_k=\mathcal{Y}_k(t^\ast_{k+1}-t^\ast_k)$ is determined
	from the sampling timestamps, $\alpha_k$ is the local
	derivative-bound parameter, and $v_k$ represents residual modeling error.
	Upon successful delivery, $\epsilon^{\mathrm{true}}_{k+1}$ from
	\eqref{eqn_TrueError_k} becomes available, and $\alpha_k$ is updated using a normalized least mean squares (NLMS) filter
	\cite{haykin2008adaptive}:
	\begin{align}
        \begin{array}{cc}
		\alpha_{k+1} = \alpha_k + \mu
		\frac{\mathcal{Y}^\ast_k}{(\mathcal{Y}^\ast_k)^2+\gamma}
		\left(\epsilon^{\mathrm{true}}_{k+1}
		-\alpha_k\mathcal{Y}^\ast_k\right),
        \end{array}
	\end{align}
	where $\mu\!\in\!(0,2)$ is the adaptation step size and $\gamma>0$ is a small parameter that prevents numerical singularity as
	$\mathcal{Y}^\ast_k\to0$.

	\section{Episodic Control}
	\label{Sec_EpisodicControl}
	
	When sample-wise control is impractical, we adopt an episodic approach that fixes the sampling interval at $h_r = Z_r T_C$ for each episode of $D$ samples, with $h_r$ determined by an $N$-th-order real-time predictor, where $N \ll D$.
	Within episode $r$, each sample is transmitted for at most $Z_r$ attempts; samples not delivered within this limit are dropped. 
	Let $\tau_k=(t'_k)^-$ denote the instant immediately before the delivery of the
$k$-th sample. The prediction error at $\tau_k$ is $e_\textsc{p}(\tau_k) \!=\! \frac{\vert{}x^{(N+1)}(\zeta_{k,r})\vert{}}
		{(N+1)!}Y_k$:
	\begin{align} 
    \begin{array}{cc}
		Y_k=\prod_{j=0}^{N}(\tau_k\!-\!t_{k-j-1}^\ast)=\!\prod_{j=0}^{N} \Delta_{k-j-1}(\tau_k),
        \end{array}
	\end{align}
	where $\zeta_{k,r}\!\in\!(t_{k-N-1}^\ast,\tau_k)$ is unknown, and $\Delta_{k-j}(t)\!=\!t\!-\!t_{k-j}^\ast$ denotes the AoI of the $(k\!-\!j)$-th delivered sample.
	We consider the expected peak prediction error over episode $r$,
	$\epsilon_r=\mathbb{E}\!\left[\max_k e_\textsc{p}(\tau_k)\right]$, and choose
	$h_r$ such that $\epsilon_r\leq\epsilon_d$. As before, we parameterize
	$\frac{\vert x^{(N+1)}(\zeta_{k,r})\vert}{(N\!+\!1)!}$ by an adaptive
	coefficient $\alpha_r$, updated at the end of each episode using the observed
	true peak error
	$\epsilon^{\mathrm{true}}_r=\max_k\epsilon^{\mathrm{true}}_k$. Thus,
	\begin{align}
        \begin{array}{cc}
		\epsilon_r
		=\mathbb{E}\!\left[\max_k e_\textsc{p}(\tau_k)\right]
		=\alpha_r\mathbb{E}\!\left[\max_k Y_k\right]
		=\alpha_r\tilde{\mathcal{Y}}_r,
        \end{array}
	\end{align}
	with the NLMS update: $\alpha_{r+1} \!=\! \alpha_r \!+\! \mu \frac{\tilde{\mathcal{Y}}_r}{(\tilde{\mathcal{Y}}_r)^2\!+\!\gamma} \!\left(\epsilon^{\mathrm{true}}_r\!\!-\!\alpha_r\tilde{\mathcal{Y}}_r\right)$. Hence, the problem reduces to deriving
	$\tilde{\mathcal{Y}}_r=\mathbb{E}\!\left[\max_k Y_k\right]$.

	\begin{lemma}
		In episode $r$ with $h_r=Z_rT_C$, the metric $Y_k$ is:
		\begin{align}
        \begin{array}{cc}
             Y_k = T_C^{N+1} \prod_{j=0}^{N}
			\big( Z_r \sum_{i=0}^{j} G_{k-i} + X_k - 1 \big),
        \end{array}
		\end{align}
		where the delivery-slot index $X_k \!\in\! \{1,\ldots,Z_r\}$ and the inter-delivery 
	sampling-slot gaps $G_m\!=\!K_m\!-\!K_{m-1} \!\ge\! 1$ have distributions 
	$\mathbb{P}(X_k\!=\!x)\!=\!p_s(1\!-\!p_s)^{x-1}/\rho$, and 
	$\mathbb{P}(G_m\!=\!g)\!=\!q^{g-1}\rho$, $g \!\geq\! 1$. Here, $q \!=\! (1 \!-\! p_s)^{Z_r}$, $\rho \!=\! 1 \!-\! q$, and $K_k$ denotes the sampling-slot index of the $k$-th delivered sample (Fig. \ref{fig_episodic_vars}).
	\end{lemma}
	\begin{proof}[Proof Sketch]
		Relative to the episode start time, the $(k\!-\!j\!-\!1)$-th delivered sample is generated at 
	$t^\ast_{k-j-1}=\big[(K_{k-j-1}\!-\!1)Z_r\!+\!1\big]T_C$, while the $k$-th sample 
	is delivered at $\tau_k=\big[(K_k\!-\!1)Z_r\!+\!X_k\big]T_C$. Thus, 
	$\tau_k-t^\ast_{k-j-1} 
	\!=\!T_C\big[Z_r(K_k\!-\!K_{k-j-1})\!+\!X_k\!-\!1\big]$. 
	Since $K_k\!-\!K_{k-j-1}\!=\!\sum_{i=0}^{j}G_{k-i}$, taking the product 
	over $j=0,\ldots,N$ yields the expression for $Y_k$. 
	Conditioned on success within a sampling slot, $X_k$ follows a truncated 
	geometric distribution with normalization 
	$\rho$, while the gaps $G_m$ are i.i.d. shifted geometric with success probability $\rho$. 
	\end{proof}
	
	\begin{theorem}
		For $D \!\gg\! N$,  $\tilde{\mathcal{Y}}_r$ is approximated by:
		\begin{align}
			\tilde{\mathcal{Y}}_r
			&\!\approx
			\mathbb{P}(\mathcal{G}\!=\!1)\mathbb{E}_{\mathcal{X}}[Y] \\
			&\!+\!
			T_C^{N+1} \! \sum\nolimits_{g=2}^{g_{\mathrm{cut}}}
			\!\mathbb{P}(\mathcal{G}\!=\!g)
			\! \prod\nolimits_{j=0}^{N}
			\!\!\Big[
			Z_r\!(g\!+\!j\bar{G}_{|g})
			\!+\!\mathbb{E}[X\!-\!1]
			\Big], \notag
		\end{align}
		where
		$\mathbb{P}(\mathcal{G}\!=\!g)
		\!=\!(1\!-\!q^g)^M\!-\!(1\!-\!q^{g-1})^M$, 
		$\bar{G}_{|g}\!\triangleq\mathbb{E}[G\!\mid\! G \!\le\! g]
		\!=\!\frac{1}{\rho}-\frac{gq^g}{1-q^g}$,
		and the zero-drop expectation is
		$
			\mathbb{E}_{\mathcal{X}}[Y]
			\!=\!T_C^{N+1}\!\sum_{x=1}^{Z_r}
			\!\!\Big(F_{\mathcal{X}}^D(x)\!-\!F_{\mathcal{X}}^D(x\!-\!1)\Big)
			\!\prod_{j=0}^{N}
			\!\Big((j\!+\!1)Z_r\!+x\!-\!1\Big),
		$
		with CDF
		$F_{\mathcal{X}}(x)=\frac{1-(1-p_s)^x}{\rho}$ for $x\in\{1,\dots,Z_r\}$.
	\end{theorem}
	\begin{proof}[Proof Sketch]
		The number of samples delivered in an episode concentrates around $M\!\approx \!D\rho$ for $D\!\gg\!N$. We condition on the worst inter-sample gap $\mathcal{G}\!\triangleq\!\max_{1\le m\le M}G_m$, whose PMF is obtained from order statistics.
		For $g=1$ (no drops), all $G_m\!=\!1$, and $Y_k$ increases with $X_k$; hence the peak is governed by $\mathcal{X}\!=\!\max_m X_m$, giving the first term. For $g\!\ge\!2$, the peak is associated with the worst gap $\mathcal{G}\!=\!g$, while the preceding gaps are bounded by $g$ and have a conditional mean $\bar{G}_{|g}$. The approximation replaces the random terms inside the product with their means and truncates the tail $\mathbb{P}(\mathcal{G}\!>\!g_{\mathrm{cut}})$.
	\end{proof}
    \vspace{-4pt}
	For a perfect channel ($p_s\!=\!1$), $X_k\!\equiv\!1$ and $\tilde{\mathcal{Y}}_r$ reduces to the deterministic lower bound $\tilde{\mathcal{Y}}_r\!=\!(N\!+\!1)! h_r^{N+1}$ (see Fig. \ref{fig_Ymaxk_D40}).

    \begin{figure}[!t]
    \centering
    \begin{minipage}[t]{0.32\columnwidth}
        \centering
        \vspace{0pt}\includegraphics[width=\linewidth]{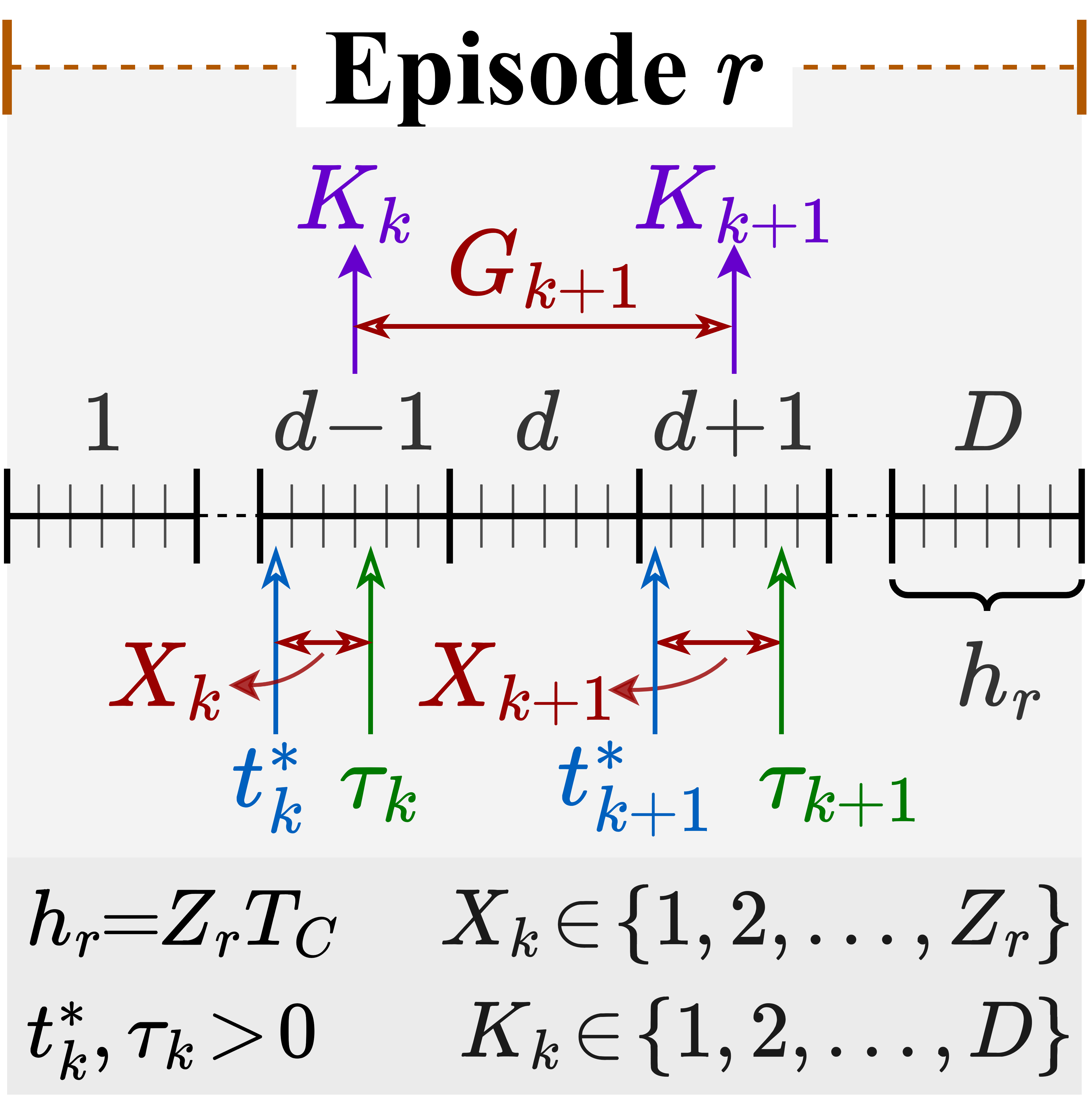}
        \caption{Variables in Lemma~1.}
        \label{fig_episodic_vars}
    \end{minipage}%
    \hfill
    \begin{minipage}[t]{0.64\columnwidth}
        \centering
        \vspace{0pt}
        \includegraphics[width=\linewidth]{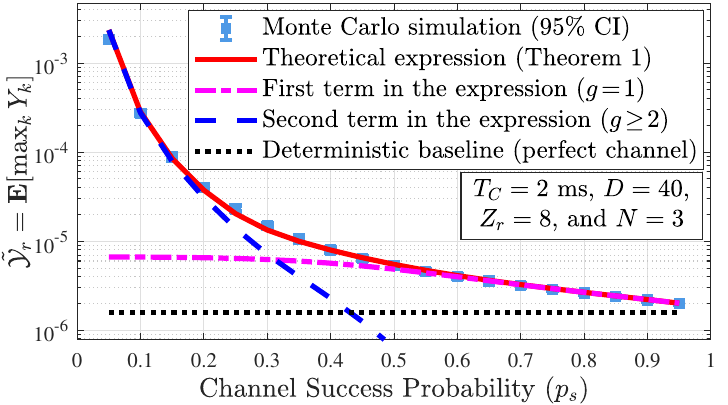}
        \caption{Theoretical $\tilde{\mathcal{Y}}_r$ vs.\ simulation.}
        \label{fig_Ymaxk_D40}
    \end{minipage}
    \end{figure}

    \begin{table*}[!tb]
		\centering
        \addtocounter{table}{2}
		\caption{Communication count and ratio to Uniform across $\bar{\epsilon}_{\textsc{r}}$ and $p_s$.}
		\label{tab_compact_comparison}
        \addtocounter{table}{-3}
		\footnotesize 
		\setlength{\tabcolsep}{1.2pt} 
		\renewcommand{\arraystretch}{0.80} 
		\newcommand{\cval}[2]{#1\,{\scriptsize\textcolor{black!70}{(#2\%)}}}
		\resizebox{\linewidth}{!}{%
			\begin{tabular}{l|*{4}{c}|*{4}{c}|*{4}{c}}
				\toprule
                \multicolumn{1}{c|}{($N=3$)}
				& \multicolumn{4}{c|}{$p_s=1.0$}
				& \multicolumn{4}{c|}{$p_s=0.75$}
				& \multicolumn{4}{c}{$p_s=0.50$} \\
				Policy \quad \ \quad \hfill $\bar{\epsilon}_{\text{\textsc{r}}}$
				& 0.01 & 0.05 & 0.10 & 0.20
				& 0.01 & 0.05 & 0.10 & 0.20
				& 0.01 & 0.05 & 0.10 & 0.20 \\
				\midrule
				\textsc{ed}
				& \cval{5144}{79} & \cval{2239}{75} & \cval{1598}{73} & \cval{980}{60}
				& \cval{6752}{76} & \cval{3035}{75} & \cval{2134}{73} & \cval{1403}{64}
				& \cval{9365}{66} & \cval{4600}{74} & \cval{3259}{73} & \cval{2285}{69} \\
				\textsc{apjsc}
				& \cval{5588}{86} & \cval{2600}{87} & \cval{1872}{86} & \cval{1381}{85}
				& \cval{7827}{88} & \cval{3524}{88} & \cval{2543}{87} & \cval{1861}{85}
				& \cval{12494}{89} & \cval{5442}{88} & \cval{3865}{86} & \cval{2836}{85} \\
				Episodic($15$)
				& \cval{6085}{93} & \cval{2839}{95} & \cval{2053}{94} & \cval{1506}{93}
				& \cval{8695}{98} & \cval{3857}{96} & \cval{2751}{94} & \cval{2027}{93}
				& \cval{13432}{95} & \cval{5981}{97} & \cval{4194}{94} & \cval{3107}{93} \\
				Episodic($20$)
				& \cval{6247}{96} & \cval{2873}{96} & \cval{2093}{96} & \cval{1544}{95}
				& \cval{8842}{99} & \cval{3967}{99} & \cval{2820}{96} & \cval{2105}{96}
				& \cval{13804}{98} & \cval{6094}{98} & \cval{4322}{96} & \cval{3180}{95} \\
				Uniform
				& \cval{6527}{100} & \cval{2982}{100} & \cval{2176}{100} & \cval{1626}{100}
				& \cval{8903}{100} & \cval{4023}{100} & \cval{2925}{100} & \cval{2185}{100}
				& \cval{14100}{100} & \cval{6194}{100} & \cval{4484}{100} & \cval{3332}{100} \\
				\bottomrule
			\end{tabular}%
		}
        \vspace{-6pt}
	\end{table*}
	
	\section{Numerical Results}
	We first assess the impact of channel awareness on the reconstruction error of \textsc{ed} and \textsc{apjsc}. Next, we examine the effect of $N$ using a linear chirp signal. Finally, we compare all policies using random stepped-frequency signals. 
	To model smoothly varying dynamics, we employ a \textbf{linear chirp} $x(t) = \cos\big(2\pi \big(f_0 + \frac{f_1-f_0}{2T_{\text{span}}}t\big)t\big)$ for $t\in[0,T_{\text{span}}]$, with $f_0=-4\text{ Hz}$, $f_1=4\text{ Hz}$, and $T_{\text{span}}=20\text{ s}$ (Fig.~\ref{fig_recon}).
    To capture diverse, randomly varying dynamics, we use \textbf{stepped-frequency signals} $y(t)=\cos(2\pi f_m t+\phi_m)$ for $t\in[(m-1)\Delta t,m\Delta t)$, where each $f_m$ is uniformly drawn from $[0,4]\text{ Hz}$ over $J=10$ intervals, with $T_{\text{span}}=80\text{ s}$ and $\Delta t=T_{\text{span}}/J=8\text{ s}$ (Fig.~\ref{fig_SteppedFreqSignal}).
	
	We evaluate the total transmission count over $T_{\text{span}}$ and the
	average peak reconstruction error:
	$
	\bar{\epsilon}_{\textsc{r}} = \frac{1}{\mathcal{K}} \sum_{k=1}^{\mathcal{K}} \epsilon^{\mathrm{true}}_k,
	$
	where $\epsilon^{\mathrm{true}}_k$ is defined in
	\eqref{eqn_TrueError_k} and $\mathcal{K}$ is the number of delivered samples.
	The simulation parameters are $T_C=2\textrm{ ms}$,
	$\mu_{\textsc{apjsc}}=0.3$, $\mu_{\textrm{Episodic}}=0.5$,
	$\gamma=10^{-15}$, and $g_{\mathrm{cut}}=100$.

	\begin{figure}[!tb]
		\centering
		\newcommand{\subfigrow}[2]{%
			\begin{minipage}[c]{0.03\linewidth}
				\subcaption{}\label{#1}
			\end{minipage}\hfill
			\begin{minipage}[c]{0.96\linewidth}
				\includegraphics[trim=0 0 4pt 0, clip,width=0.97\linewidth]{#2}
			\vspace{-0.5pt}
            \end{minipage}
		}
		
		\subfigrow{fig_recon}{ReconstructedSignals_NU.eps}
		\subfigrow{fig_err_u}{Error_Uniform.eps}
		\subfigrow{fig_err_a}{Error_AdaptiveJSC.eps}
		\subfigrow{fig_ep_error}{Error_Episodic.eps}
		\subfigrow{fig_err_e}{Error_EventDriven.eps}
		\subfigrow{fig_times}{SamplingTimes.eps}
		
		\caption{Original and reconstructed signals across policies (\subref{fig_recon}), reconstruction errors (\subref{fig_err_u})--(\subref{fig_err_e}), and sampling intervals (\subref{fig_times}).}
		\label{fig_all_signals}
        \vspace{2pt}
	\end{figure}
	
	\begin{figure}[!tb]
		\centering
		\includegraphics[trim=0 0pt 0 0, clip, width=0.96\linewidth]{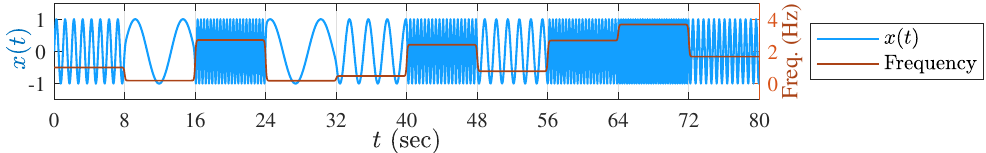}
		\caption{A realization of the random stepped-frequency signal.}
		\label{fig_SteppedFreqSignal}
	\end{figure}
	
	\textbf{Channel-Aware Performance:}
	Table~\ref{tab_comparison_Ch} evaluates the channel compensation factor
	$\mathcal{C}$ in \eqref{eqn_C_coef} across erasure probabilities for the
	chirp signal. In the channel-agnostic setting ($\mathcal{C}=1$),
	$\bar{\epsilon}_{\textsc{r}}$ increases as channel quality degrades.
	For \textsc{ed} with $\epsilon_d=0.01$ and $p_s=0.5$,
	$\bar{\epsilon}_{\textsc{r}}=0.0142$ violates the design budget by $42\%$.
	In contrast, incorporating $\mathcal{C}$ keeps $\bar{\epsilon}_{\text{\textsc{r}}}$ well within the target (up to slight variations due to the approximation in~\eqref{eqn_ch_approx}) by scheduling proactive transmissions to prevent violations.

\begin{table}[!t]
    \centering
    \caption{Error $\bar{\epsilon}_{\text{\textsc{r}}}$ under channel-aware scheduling.}
    \label{tab_comparison_Ch}
    \scriptsize
    \setlength{\tabcolsep}{3.0pt}
    \renewcommand{\arraystretch}{0.85}
    \setlength{\aboverulesep}{1.0pt}
    \setlength{\belowrulesep}{1.6pt}
    \setlength{\cmidrulesep}{1.0pt}
    \begin{tabular}{ll ccc ccc}
        \toprule
        & ($N=3$) & \multicolumn{3}{c}{\textsc{ed}} & \multicolumn{3}{c}{\textsc{apjsc}} \\
        \cmidrule(lr){3-5} \cmidrule(lr){6-8}
        $p_s$ & Design  /  $\epsilon_d$ & 0.01 & 0.1 & 0.2 & 0.01 & 0.1 & 0.2 \\
        \midrule
        1.0 & Ideal    & 0.0112 & 0.1055 & 0.2156 & 0.0123 & 0.1110 & 0.2171 \\
        \arrayrulecolor{black!50}\cmidrule(lr){1-8}\arrayrulecolor{black} 
        \multirow{2}{*}{0.75}
            & Aware    & 0.0104 & 0.1047 & 0.2210 & 0.0113 & 0.1082 & 0.2139 \\
            & Agnostic & 0.0121 & 0.1090 & 0.2267 & 0.0130 & 0.1129 & 0.2206 \\
        \arrayrulecolor{black!50}\cmidrule(lr){1-8}\arrayrulecolor{black} 
        \multirow{2}{*}{0.50}
            & Aware    & \textbf{0.0093} & \textbf{0.1038} & \textbf{0.2299} & \textbf{0.0090} & \textbf{0.1025} & \textbf{0.2084} \\
            & Agnostic & 0.0142 & 0.1176 & 0.2488 & 0.0145 & 0.1168 & 0.2278 \\
        \bottomrule
    \end{tabular}
    \vspace{2pt}
\end{table}

	\textbf{Impact of Prediction Order $N$:} 
	Table~\ref{tab_comparison_N} reports transmission counts for different
	prediction orders under $p_s=0.75$. Increasing $N$ from $1$ to $3$ yields
	substantial savings (e.g., \textsc{apjsc} decreases from $2852$ to $1927$
	transmissions for $\bar{\epsilon}_{\textsc{r}}=0.01$). Beyond $N=3$, gains
	plateau or slightly reverse due to polynomial
	oscillations over wider prediction intervals (we fix $N=3$ thereafter).
	\textsc{ed} achieves the lowest transmission count but requires continuous
	fine-grained sensing at intervals of $T_C$
	($T_{\text{span}}/T_C=10{,}000$ samples), providing no sampling savings.
	
	Figure~\ref{fig_all_signals} shows the reconstructed signals
	$\hat{x}(t)$, errors $e_{\textsc{R}}(t)$, and sampling intervals for the
	chirp under a fixed budget of $\sim\!750$ transmission attempts over
	$20\text{ s}$ with $p_s=0.75$. \textsc{ed} maintains nearly constant error (Fig. \ref{fig_err_e}) by
	triggering transmissions when the true error exceeds $\epsilon_d/\mathcal{C}$,
	yielding $\bar{\epsilon}_{\textsc{r}}=0.05$. The uniform policy has near-zero
	error during slow variations but exceeds $0.4$ during rapid variations (Fig. \ref{fig_err_u}),
	yielding $\bar{\epsilon}_{\textsc{r}}=0.08$. \textsc{apjsc} adapts its sampling
	interval to local dynamics, increasing it during slow variations
	(around $t=10\text{ s}$) and decreasing it during rapid variations (Fig.~\ref{fig_times}), with
	$\bar{\epsilon}_{\textsc{r}}=0.06$. The episodic policy updates its interval
	once per episode ($D=20$) and thus adapts more slowly than \textsc{apjsc} but faster than the
	uniform policy (Fig.~\ref{fig_times}), achieving
	$\bar{\epsilon}_{\textsc{r}}=0.078$.

    \begin{table}[!t]
    \centering
    \caption{Communication comparison across $N$ {\small($95\%$ CI)}.}
    \label{tab_comparison_N}
    \footnotesize
    \setlength{\tabcolsep}{1.5pt}
    \renewcommand{\arraystretch}{0.88}
    \newcommand{\ci}[1]{{\tiny\textcolor{black!70}{$\pm#1$}}}
    \resizebox{\columnwidth}{!}{%
    \begin{tabular}{c cccc cccc}
        \toprule
        \rule{0pt}{1.6ex}& \multicolumn{2}{c}{\textsc{ed}} & \multicolumn{2}{c}{\textsc{apjsc}} & \multicolumn{2}{c}{Episodic} & \multicolumn{2}{c}{Uniform} \\[-0.8ex]
        \cmidrule(lr){2-3}\cmidrule(lr){4-5}\cmidrule(lr){6-7}\cmidrule(lr){8-9}
        \rule{0pt}{1.5ex}$N$ \ / \ $\bar{\epsilon}_{\text{\textsc{r}}}$ & 0.01 & 0.1 & 0.01 & 0.1 & 0.01 & 0.1 & 0.01 & 0.1 \\[-0.5ex]
        \midrule
        1 & 2500\ci{43} & 737\ci{31} & 2852\ci{55} & 867\ci{33} & 3132\ci{96}  & 981\ci{33} & 3288\ci{69} & 981\ci{37} \\
        2 & 1764\ci{39} & 550\ci{27} & 1987\ci{49} & 679\ci{27} & 2117\ci{104} & 784\ci{35} & 2305\ci{61} & 791\ci{31} \\
        3 & 1628\ci{41} & 535\ci{29} & 1927\ci{49} & 585\ci{22} & 2073\ci{120} & 671\ci{33} & 2185\ci{53} & 671\ci{29} \\
        4 & 1508\ci{45} & 443\ci{27} & 1926\ci{49} & 584\ci{29} & 2070\ci{137} & 661\ci{41} & 2200\ci{57} & 606\ci{29} \\
        5 & 1552\ci{43} & 446\ci{33} & 1919\ci{47} & 605\ci{25} & 2082\ci{131} & 657\ci{37} & 2192\ci{59} & 655\ci{29} \\
        \bottomrule
    \end{tabular}%
    }
\end{table}

	\textbf{Communication and Sampling Efficiency:} 
	\label{sec_SteppedFreResults}
	Table~\ref{tab_compact_comparison} reports the transmission count required to
	achieve specific reconstruction errors $\bar{\epsilon}_{\textsc{r}}$ for
	different policies and $p_s$ values on stepped-frequency signals. Results are
	averaged over $200$ Monte Carlo iterations, each with 10 random frequencies.
	Relative to the uniform policy: (i)~\textbf{\textsc{ED}} reduces transmissions by
	$21\%\text{--}40\%$, providing a communication lower bound at the cost of fine-grained sensing; (ii)~\textbf{\textsc{APJSC}} reduces transmissions by
	$11\%\text{--}15\%$ while substantially lowering sampling overhead by \emph{sampling
	only when required}; and
	(iii)~\textbf{Episodic} achieves reductions of $2\%\text{--}7\%$ ($D=15$) and
	$1\%\text{--}5\%$ ($D=20$), showing that episodic control can also improve
	communication efficiency for coarse-grained control hardware.

    \vfill\pagebreak
    \section{Acknowledgment}
    
    This work has been supported by ELLIIT and the European Union (\textsc{ELIXIRION}, 101120135).
    
	\bibliographystyle{IEEEbib}
	\bibliography{Refs}

@inproceedings{AstromBernhardsson2002,
  title={Comparison of Riemann and Lebesgue sampling for first order stochastic systems},
  author={Karl J. {\AA}str{\"o}m and Bo Bernhardsson},
  booktitle={Proceedings of the 41st IEEE Conference on Decision and Control},
  volume={2},
  pages={2011--2016},
  year={2002},
  organization={}
}

@article{Tabuada2007,
  title={Event-triggered real-time scheduling of stabilizing control tasks},
  author={Tabuada, Paulo},
  journal={IEEE Transactions on Automatic control},
  volume={52},
  number={9},
  pages={1680--1685},
  year={2007},
  publisher={}
}

@article{AntaTabuada2010,
  title={To sample or not to sample: Self-triggered control for nonlinear systems},
  author={Anta, Adolfo and Tabuada, Paulo},
  journal={IEEE Transactions on automatic control},
  volume={55},
  number={9},
  pages={2030--2042},
  year={2010},
  publisher={}
}

@inproceedings{HeemelsJohanssonTabuada2012,
  title={An introduction to event-triggered and self-triggered control},
  author={Heemels, Wilhelmus PMH and Johansson, Karl Henrik and Tabuada, Paulo},
  booktitle={Proceedings of the 51st IEEE Conference on Decision and Control},
  pages={3270--3285},
  year={2012},
  organization={}
}

@article{Miskowicz2006,
  author  = {Marek Mi{\'s}kowicz},
  title   = {Send-On-Delta Concept: An Event-Based Data Reporting Strategy},
  journal = {Sensors},
  volume  = {6},
  number  = {1},
  pages   = {49--63},
  year    = {2006},
  doi     = {10.3390/s6010049}
}

@inproceedings{FeiziGoyalMedard2010,
  author    = {Soheil Feizi and Vivek K. Goyal and Muriel M{\'e}dard},
  title     = {Locally Adaptive Sampling},
  booktitle = {Proceedings of the 48th Annual Allerton Conference on
               Communication, Control, and Computing},
  pages     = {152--159},
  year      = {2010},
  doi       = {10.1109/ALLERTON.2010.5706901}
}

@article{FeiziGoyalMedard2012,
  author  = {Soheil Feizi and Vivek K. Goyal and Muriel M{\'e}dard},
  title   = {Time-Stampless Adaptive Nonuniform Sampling for
             Stochastic Signals},
  journal = {IEEE Transactions on Signal Processing},
  volume  = {60},
  number  = {10},
  pages   = {5440--5450},
  year    = {2012},
  doi     = {10.1109/TSP.2012.2208633}
}

@article{PremaratneHalgamugeMareels2013,
  author  = {U. Premaratne and S. K. Halgamuge and I. M. Y. Mareels},
  title   = {Event Triggered Adaptive Differential Modulation:
             A New Method for Traffic Reduction in Networked Control Systems},
  journal = {IEEE Transactions on Automatic Control},
  volume  = {58},
  number  = {7},
  pages   = {1696--1706},
  year    = {2013},
  doi     = {10.1109/TAC.2013.2248015}
}

@article{Schenato2008,
  author  = {Luca Schenato},
  title   = {Optimal estimation in networked control systems subject to random delay and packet drop},
  journal = {IEEE Transactions on Automatic Control},
  volume  = {53},
  number  = {5},
  pages   = {1311--1317},
  year    = {2008},
  doi     = {10.1109/TAC.2008.921341}
}

@article{NetworkedEventBasedDropout2009,
  author  = {Nguyen, Vinh Hao and Suh, Young Soo},
  title   = {Networked estimation for event-based sampling systems with packet dropouts},
  journal = {Sensors},
  volume  = {9},
  number  = {4},
  pages   = {3078--3100},
  year    = {2009},
  doi     = {10.3390/s90403078}
}

@article{PengHan2016,
  author  = {Peng, Chen and Han, Qing-Long},
  title   = {On designing a novel self-triggered sampling scheme for networked control systems with data losses and communication delays},
  journal = {IEEE Transactions on Industrial Electronics},
  volume  = {63},
  number  = {2},
  pages   = {1239--1248},
  year    = {2016},
  doi     = {10.1109/TIE.2015.2486776}
}

@article{SunPolyanskiyUysal2020,
  author  = {Yin Sun and Yury Polyanskiy and Elif Uysal},
  title   = {Sampling of the Wiener process for remote estimation over a channel with random delay},
  journal = {IEEE Transactions on Information Theory},
  volume  = {66},
  number  = {2},
  pages   = {1118--1135},
  year    = {2020},
  doi     = {10.1109/TIT.2019.2948616}
}

@article{kosta2017age,
  title={Age of information: A new concept, metric, and tool},
  author={Kosta, Antzela and Pappas, Nikolaos and Angelakis, Vangelis},
  journal={Foundations and Trends in Networking},
  volume={12},
  number={3},
  pages={162--259},
  year={2017},
  publisher={Emerald Publishing Limited}
}

@article{YatesSunBrownKaulModianoUlukus2021,
  author  = {Yates, Roy D and Sun, Yin and Brown, D Richard and Kaul, Sanjit K and Modiano, Eytan and Ulukus, Sennur},
  title   = {Age of information: An introduction and survey},
  journal = {IEEE Journal on Selected Areas in Communications},
  volume  = {39},
  number  = {5},
  pages   = {1183--1210},
  year    = {2021},
  doi     = {10.1109/JSAC.2021.3065072}
}

@inproceedings{shisher2024monotonicity,
  title={On the monotonicity of information aging},
  author={Shisher, MD Kamran Chowdhury and Sun, Yin},
  booktitle={IEEE INFOCOM 2024-IEEE Conference on Computer Communications Workshops (INFOCOM WKSHPS)},
  pages={01--06},
  year={2024},
  organization={}
}

@article{ornee2021sampling,
  title={Sampling and remote estimation for the Ornstein-Uhlenbeck process through queues: Age of information and beyond},
  author={Ornee, Tasmeen Zaman and Sun, Yin},
  journal={IEEE/ACM Transactions on Networking},
  volume={29},
  number={5},
  pages={1962--1975},
  year={2021},
  publisher={}
}

@article{kountouris2021semantics,
  title={Semantics-empowered communication for networked intelligent systems},
  author={Kountouris, Marios and Pappas, Nikolaos},
  journal={IEEE Communications Magazine},
  volume={59},
  number={6},
  pages={96--102},
  year={2021},
  publisher={}
}

@article{luo2025information,
  title={From information freshness to semantics of information and goal-oriented communications},
  author={Luo, Jiping and Delfani, Erfan and Salimnejad, Mehrdad and Pappas, Nikolaos},
  journal={arXiv preprint arXiv:2512.12758},
  year={2025}
}

@book{haykin2008adaptive,
  title={Adaptive filter theory},
  author={Haykin, Simon S},
  year={2008},
  publisher={Pearson Education India}
}

@book{stoer1993introduction,
  author    = {Stoer, Josef and Bulirsch, Roland},
  title     = {Introduction to Numerical Analysis},
  edition   = {2nd},
  publisher = {Springer-Verlag},
  address   = {New York},
  year      = {1993}
}

@inproceedings{kutsevol2023experimental,
  title={Experimental study of transport layer protocols for wireless networked control systems},
  author={Kutsevol, Polina and Ayan, Onur and Pappas, Nikolaos and Kellerer, Wolfgang},
  booktitle={2023 20th Annual IEEE International Conference on Sensing, Communication, and Networking (SECON)},
  pages={438--446},
  year={2023},
  organization={}
}

\end{document}